\documentclass[journal]{IEEEtran}
\usepackage{amssymb,stmaryrd,amsmath,amsfonts,rotating}
\usepackage[noadjust]{cite} \usepackage{color} 
\usepackage[vflt]{floatflt} \usepackage{epic}
\usepackage{color}

\newcommand{\x}{\underline{x}}

\RequirePackage{bbm} 
\definecolor{TODO}{rgb}{0.6,0.6,0.6} 

\definecolor{TOCHECK}{rgb}{0.8,0.8,0.8} 

\newtheorem{theorem}{Theorem}
\newcommand{\btheo}{\begin{theorem}}
\newcommand{\etheo}{\end{theorem}}
\newcommand{\bproof}{\begin{proof}}
\newcommand{\eproof}{\end{proof}}
\newtheorem{definition}[theorem]{Definition}
\newcommand{\bdefi}{\begin{definition}}
\newcommand{\edefi}{\end{definition}}
\newtheorem{fact}[theorem]{Fact}
\newcommand{\bprop}{\begin{fact}}
\newcommand{\eprop}{\end{fact}}
\newtheorem{corollary}[theorem]{Corollary}
\newcommand{\bcor}{\begin{corollary}}
\newcommand{\ecor}{\end{corollary}}
\newtheorem{example}[theorem]{Example}
\newcommand{\bex}{\begin{example}}
\newcommand{\eex}{\end{example}}
\newtheorem{lemma}[theorem]{Lemma}
\newcommand{\blemma}{\begin{lemma}}
\newcommand{\elemma}{\end{lemma}}
\newtheorem{remark}[theorem]{Remark}
\newcommand{\bremark}{\begin{remark}}
\newcommand{\eremark}{\end{remark}}
\newtheorem{conj}[theorem]{Conjecture}
\newcommand{\bconj}{\begin{conj}}
\newcommand{\econj}{\end{conj}}

\newcommand{\naturals}{\ensuremath{\mathbb{N}}}
\newcommand{\integers}{\ensuremath{\mathbb{Z}}}

\def\0{{\tt 0}} 
\def\1{{\tt 1}} 
\def\?{{\tt *}} 
\DeclareMathOperator{\w}{w}    
\newcommand{\BPsmall}{\ensuremath{\text{\tiny BP}}} 
\newcommand{\qed}{{\hfill \footnotesize $\blacksquare$}}

\newcommand{\dens}[1]{\mathsf{#1}}
\newcommand{\Ldens}[1]{\dens{#1}}

\newcommand{\dr}{{\mathtt r}}
\newcommand{\dl}{{\mathtt l}}

\newcommand{\vconv}{\circledast}
\newcommand{\cconv}{\boxast}

\allowdisplaybreaks
\begin{document} 
\title{
Threshold Saturation on BMS Channels via Spatial Coupling} \author{\authorblockN{Shrinivas
Kudekar\authorrefmark{1}, Cyril M{\'e}asson\authorrefmark{2}, Tom Richardson\authorrefmark{2} and R{\"u}diger
Urbanke\authorrefmark{3} \\ }
\authorblockA{\authorrefmark{1} New Mexico Consortium  and CNLS, Los Alamos National Laboratory, New Mexico, USA\\ Email: skudekar@lanl.gov} \\
\authorblockA{\authorrefmark{2} Qualcomm, USA\\ Email: \{tjr, measson\}@qualcomm.com} \\
 \authorblockA{\authorrefmark{3}School of
Computer and Communication Sciences, EPFL, Lausanne, Switzerland\\
Email: ruediger.urbanke@epfl.ch}\\
 }

\maketitle
\begin{abstract}
We consider spatially coupled code ensembles. A particular instance
are convolutional LDPC ensembles.  It was recently shown that, for
transmission over the binary erasure channel, this coupling increases
the belief propagation threshold of the ensemble to the maximum
a-priori threshold of the underlying component ensemble.  We report
on empirical evidence which suggest that the same phenomenon also
occurs when transmission takes place over a general binary memoryless
symmetric channel. This is confirmed both by simulations as well as by computing EBP GEXIT
curves and by comparing the empirical BP thresholds of coupled ensembles
to the empirically determined MAP thresholds of the underlying regular
ensembles.

We further consider ways of reducing the rate-loss incurred by such
constructions.  \end{abstract}

\section{Introduction}
It has long been known that convolutional LDPC ensembles, introduced by
Felstr{\"{o}}m and Zigangirov \cite{FeZ99}, have excellent thresholds
when transmitting over general binary-input symmetric-output memoryless
(BMS) channels. The fundamental reason underlying this good performance
was recently discussed in detail in \cite{KRU10} for the case when
transmission takes place over the binary erasure channel (BEC).

In particular, it was shown in \cite{KRU10} that the BP threshold of the
spatially coupled ensemble is essentially equal to the MAP threshold
of the underlying component ensemble. It was also shown that for long
chains the MAP performance of the chain cannot be substantially larger
than the MAP threshold of the component ensemble.  In this sense, the
BP threshold of the chain is increased to its maximal possible value.
This is the reason why we call this phenomena {\em threshold saturation
via spatial coupling}. In a recent paper \cite{LeF10}, Lentmaier and
Fettweis independently formulated the same statement as conjecture. They
attribute the observation of the equality of the two thresholds
to G. Liva.

It is tempting to conjecture that the same phenomenon occurs for
transmission over general BMS channels. We provide some empirical
evidence that this is indeed the case. In particular, we compute EBP
GEXIT curves for transmission over the binary additive white Gaussian
noise (BAWGN) channel.  We show that these curves behave in an identical
fashion to the ones when transmission takes place over the BEC. We also
compute fixed points (FPs) of the spatial configuration and we demonstrate
again empirically that these FPs have properties
identical to the ones in the BEC case.

For a review on the literature on convolutional LDPC ensembles we refer
the reader to \cite{KRU10} and the references therein.  As discussed in
\cite{KRU10}, there are many basic variants of coupled ensembles.  For the
sake of convenience of the reader, we quickly review the ensemble $(\dl,
\dr, L, w)$. This is the ensemble we use throughout the paper.

We assume that the variable nodes are at positions $[-L, L]$, $L
\in \naturals$. At each position there are $M$ variable nodes, $M
\in \naturals$. Conceptually we think of the check nodes to be
located at all integer positions from $[- \infty, \infty]$.  Only
some of these positions actually interact with the variable nodes.
At each position there are $\frac{\dl}{\dr} M$ check nodes. It
remains to describe how the connections are chosen.  We assume that
each of the $\dl$ connections of a variable node at position $i$
is uniformly and independently chosen from the range $[i, \dots,
i+w-1]$, where $w$ is a ``smoothing'' parameter. In the same way,
we assume that each of the $\dr$ connections of a check node at
position $i$ is independently chosen from the range $[i-w+1, \dots,
i]$.

A discussion on the above ensemble and a proof of the following lemma can be found in \cite{KRU10}.
\begin{lemma}[Design Rate]\label{lem:designrate}
The design rate of the ensemble $(\dl, \dr, L, w)$, with $w \leq 2 L$,
is given by
\begin{align*}
R(\dl, \dr, L, w) & = 
(1-\frac{\dl}{\dr}) - \frac{\dl}{\dr} \frac{w+1-2\sum_{i=0}^{w} 
\bigl(\frac{i}{w}\bigr)^{\dr}}{2 L+1}.
\end{align*}
\end{lemma}

\section{Review: EBP GEXIT Curves, the Area Theorem, and the Maxwell Construction}
\label{sec:review}
Our aim is to empirically demonstrate that the performance of coupled
ensembles is closely related to that of the underlying ensemble also
in the general case. We limit our discussion to the coupling of regular
ensembles. To get started, let  us briefly review how the BP and
MAP threshold can be characterized for regular ensembles. A detailed discussion can be found
in \cite{RiU08}.

For $\ell \geq 1$, the (forward) density evolution (DE) equation for an
$(\dl, \dr)$-regular ensemble is given by
$$
\Ldens{x}_{\ell} = \Ldens{c} \vconv (\Ldens{x}^{\cconv (\dr-1)}_{\ell-1})^{\vconv (\dl-1)}.
$$
Here, $\Ldens{c}$ is the $L$-density of the BMS channel over which
transmission takes place and $\Ldens{x}_{\ell}$ is the density emitted
by variable nodes in the $\ell$-th round of density evolution. Initially
we have $\Ldens{x}_{0}=\Delta_0$, the delta function at $0$. The operators $\vconv$ and $\cconv$ correspond
to the convolution of densities at variable and check nodes, respectively.
We say that a density $\Ldens{x}$ is a FP of DE for the channel
$\Ldens{c}$ if
$
\Ldens{x} = \Ldens{c} \vconv (\Ldens{x}^{\cconv (\dr-1)})^{\vconv (\dl-1)}.
$
More succinctly, we say that $(\Ldens{c}, \Ldens{x})$ is a FP of DE.

For the BEC it is known that the behavior of the BP as well as the behavior of the MAP
decoder are determined by the collection of FPs of DE. The same
is conjectured to be true for general channels. Let us discuss this general
conjecture.

The key concept in this conjecture is the EBP GEXIT curve \cite{MMRU09}.  This curve
is shown in Figure~\ref{fig:ebpexit36} for the $(3, 6)$-regular ensemble assuming
that transmission takes place over the BAWGN.
Numerically it is constructed
in the following way. To construct one point on this curve find a FP $(\Ldens{c}_{\sigma}, \Ldens{x})$, 
where $\Ldens{c}_{\sigma}$ denotes an element
of the channel family under consideration. E.g., in the case considered
in Figure~\ref{fig:ebpexit36}, $\Ldens{c}_{\sigma}$ represents the $L$-density
of a BAWGN channel of variance $\sigma^2$. This FP gives 
rise to the point $(H(\Ldens{c}_{\sigma}), G(\Ldens{c}_{\sigma}, \Ldens{x}))$ in the GEXIT plot. 
Hereby, 
\begin{align*}
H(\Ldens{a}) & = \int \Ldens{a}(y) \log(1+e^{-y}) dy, \\
G(\Ldens{c}_{\sigma}, \Ldens{a}) & = \int \Ldens{a}(y) l(\sigma, y) dy  \\
l(\sigma, y) & = 
\Bigl(\int \frac{e^{-\frac{(z-2/\sigma^2)^2 \sigma^2}{8}}}{1+e^{z+y}} dz \Bigr)/
\Bigl(\int \frac{e^{-\frac{(z-2/\sigma^2)^2 \sigma^2}{8}}}{1+e^{z}} dz \Bigr). 
\end{align*}
In words, $H( \cdot )$ computes the {\em entropy} associated to an $L$-density,
whereas $G(\Ldens{c}_{\sigma}, \cdot )$ computes the so-called {\em GEXIT value} of
an $L$-density. This GEXIT value depends on the ``operating point'', i.e., it depends
on the underlying channel $\Ldens{c}_{\sigma}$. To first order, the GEXIT
value is equal to the entropy. 

We get the EBP GEXIT curve if we plot the points corresponding
to {\em all} FPs of DE. For a detailed
discussion we refer the reader to \cite{MMRU09,RiU08}.
\begin{figure}[htp] \centering
\input{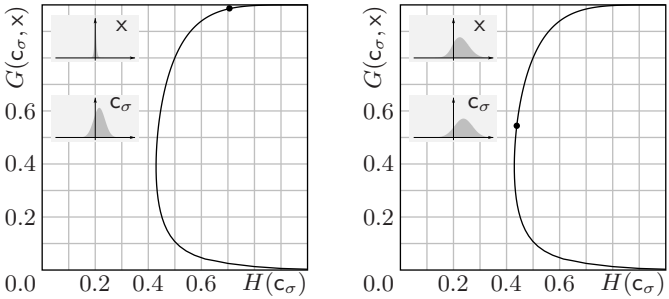} 
\caption{\label{fig:ebpexit36} 
The EBP GEXIT curve for the $(\dl=3,
\dr=6)$-regular ensemble and transmission over the BAWGNC.  Each point on
the curve corresponds to a FP $(\Ldens{c}_{\sigma}, \Ldens{x})$ of DE.  
The two figures show the 
FP density $\Ldens{x}$ 
as well as the input density $\Ldens{c}_{\sigma}$
for two points on the curves (see inlets).} 
\end{figure}
It was shown in \cite{MMRU09} that, under suitable technical conditions,
for every GEXIT value $g \in [0, 1]$ there exists at least one FP $(\Ldens{c}_{\sigma}, \Ldens{x})$ with GEXIT value $g$.  Further, a
simple recursive numerical procedure can be used to find such a FP. 
The technical difficulty lies in establishing the existence of the
{\em curve} (rather than the existence of just the set of FPs).
Although the numerical evidence strongly suggests the existence of the
curve, it is an open problem to prove this analytically.

We can construct an upper bound on the MAP threshold as shown in
Figure~\ref{fig:maxwell}, see \cite{MMRU09}: integrate the EBP GEXIT
curve starting from the point $(1, 1)$ from right to left until the area
under the curve equals the rate of the code. The point on the $x$-axis
where this equality occurs is an upper bound on the MAP threshold. It
is conjectured to be in fact the exact MAP threshold.

\begin{figure}[htp]
\centering
\input{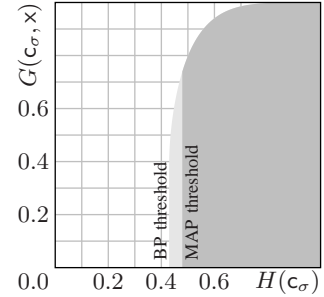}
\caption{\label{fig:maxwell} 
Upper bound on the MAP threshold for the $(3,6)$-regular ensemble and
transmission over the BAWGNC.  This upper bound is given by the entropy
value where the dark gray vertical line hits the $x$-axis.  The Maxwell
conjecture states that this bound is tight. Numerically the upper bound
is at a channel entropy of roughly $0.4792$. For comparison, the BP
threshold is at a channel entropy of roughly $0.4291$.} \end{figure}


\section{Density Evolution, Fixed Points, and the EBP GEXIT Curve for Coupled Ensembles}
\subsection{Density Evolution}
Let us describe the DE equations for the $(\dl, \dr, L, w)$ ensemble.
In the sequel, densities are $L$-densities.  Let $\Ldens{c}$ denote
the channel density and let $\Ldens{x}_i$ denote the density which is
emitted by variable nodes at position $i$.

\begin{definition}[Density Evolution of $(\dl, \dr, L, w)$ Ensemble]
Let $\Ldens{x}_i$, $i\in \integers$, denote the average $L$-density which
is emitted by variable nodes at position $i$. For $i \not \in [-L, L]$
we set $\Ldens{x}_i=\Delta_{+\infty}$. Here, $\Delta_{+\infty}$ is the
delta function at $+\infty$. In words, the boundary variable nodes have
perfect information. For $i \in [-L, L]$, the FP condition implied by
DE is
\begin{align}\label{eq:densevolxi}
\Ldens{x}_i 
& = \Ldens{c} \vconv 
\Bigl(\frac{1}{w} \sum_{j=0}^{w-1} \bigl(\frac{1}{w} \sum_{k=0}^{w-1} 
\Ldens{x}_{i+j-k} \bigr)^{\cconv (\dr-1)} \Bigr)^{\vconv(\dl-1)}.
\end{align}
Define
\begin{align*}
g(\Ldens{x}_{i\!-\!w\!+\!1}, \!\dots\!,  \Ldens{x}_{i\!+\!w\!-\!1}) =  \Bigl(\frac{1}{w} \sum_{j=0}^{w-1} \bigl(\frac{1}{w} \sum_{k=0}^{w-1} 
\Ldens{x}_{i\!+\!j\!-\!k} \bigr)^{\cconv (\dr\!-\!1)} \Bigr)^{\vconv(\dl\!-\!1)}.
\end{align*}
Note that 
$g(\Ldens{x}, \dots, \Ldens{x}) = (\Ldens{x}^{\cconv (\dr-1)})^{\vconv (\dl-1)}$,
where the r-h-s represents DE (without the effect of the
channel) for the underlying $(\dl, \dr)$-regular ensemble.

We write $\Ldens{y} \prec \Ldens{x}$ if $\Ldens{x}$ is {\em degraded}
w.r.t. $\Ldens{y}$.  It is not hard to see \cite{RiU08} that the
function $g(\Ldens{x}_{i-w+1}, \dots, \Ldens{x}_{i})$
 is monotone w.r.t. degradation in all its arguments $\Ldens{x}_j$,
$j=i-w+1, \dots, i$.  More precisely, if we degrade any of the densities
$\Ldens{x}_j$, $j=i-w+1, \dots, i$, then $g(\cdot)$  is also degraded
w.r.t. to its original value. We say that $g(\cdot)$ is {\em
monotone} in its arguments.
\qed 
\end{definition}

\subsection{Fixed Points and Admissible Schedules}
\begin{definition}[FPs of Density Evolution]\label{def:fixedpoints}
Consider DE for the $(\dl, \dr, L, w)$ ensemble.
Let $\x=(\Ldens{x}_{-L}, \dots, \Ldens{x}_{L})$. We call $\x$ the {\em constellation} (of symmetric $L$-densities). 
We say that $\x$ forms a FP
of DE with channel $\Ldens{c}$ if $\x$ fulfills (\ref{eq:densevolxi})
for $i \in [-L, L]$.  As a short hand we then say that $(\Ldens{c}, \x)$
is a FP.  We say that $(\Ldens{c}, \x)$ is a {\em non-trivial}
FP if $\x$ is not identically equal to $\Delta_{+\infty}\,\forall\, i$.
Again, for $i\notin [-L,L]$, $\Ldens{x}_i = \Delta_{+\infty}$.
\qed
\end{definition}

\begin{definition}[Forward DE and Admissible Schedules]\label{def:forwardDE} 
Consider {\em forward} DE for the $(\dl, \dr, L, w)$ ensemble.  More
precisely, pick a channel $\Ldens{c}_{\sigma}$. Initialize 
$\x^{(0)}=(\Delta_0, \dots, \Delta_0)$. Let $\x^{(\ell)}$ be the result of
$\ell$ rounds of DE. More precisely, $\x^{(\ell+1)}$ is generated from
$\x^{(\ell)}$ by applying the DE equation \eqref{eq:densevolxi} to each
section $i\in [-L,L]$,
\begin{align*}
\Ldens{x}_i^{(\ell+1)} & = \Ldens{c}\vconv g(\Ldens{x}_{i-w+1}^{(\ell)},\dots,\Ldens{x}_{i+w-1}^{(\ell)}).
\end{align*}
We call this the {\em parallel} schedule.

More generally, consider a schedule in which in each step $\ell$
an arbitrary subset of the sections is updated, constrained only by
the fact that every section is updated in infinitely many steps. We
call such a schedule {\em admissible}. Again, we call $\x^{(\ell)}$
the resulting sequence of constellations.  \end{definition}

\begin{lemma}[FPs of Forward DE]\label{lem:forwardDE} 
Consider forward DE for the $(\dl, \dr, L, w)$ ensemble.  Let
$\x^{(\ell)}$ denote the sequence of constellations under an admissible
schedule.  Then $\x^{(\ell)}$ converges to a FP of DE, with each component being a symmetric $L$-density and this
FP is independent of the schedule.  In particular, it is equal
to the FP of the parallel schedule.
\end{lemma}
\begin{proof}
Consider first the parallel schedule. We claim that the vectors
$\x^{(\ell)}$ are ordered, i.e., $\x^{(0)}\succ \x^{(1)}\succ \dots
\succ \underline{0}$ (the ordering is section-wise and $\underline{0}$ is the vector of $\Delta_{+\infty}$). This is true since
$\x^{(0)}=(\Delta_0, \dots, \Delta_0)$, whereas $\x^{(1)}\prec (\Ldens{c},
\dots, \Ldens{c}) \prec (\Delta_0, \dots, \Delta_0) =\x^{(0)}$. It now
follows by induction on the number of iterations and the monotonicity
of the function $g(\cdot)$, $\forall\,i$, that the sequence
$\x^{(\ell)}$ is monotonically decreasing. More precisely, we have
 $\x^{(\ell+1)}_i \prec \x^{(\ell)}_i$. Hence, from Lemma 4.75 in
\cite{RiU08}, we conclude that each section converges to a limit density
which is also symmetric.
 Call the limit $\x^{\infty}$.  Since the DE equations
are continuous it follows that $\x^{(\infty)}$ is a FP of DE
\eqref{eq:densevolxi} with parameter $\Ldens{c}$. We call $\x^{(\infty)}$
the forward FP of DE.

That the limit (exists in general and that it) does not depend on the
schedule follows by standard arguments and we will be brief.  The idea
is that for any two admissible schedules the corresponding computation
trees are nested. This means that if we look at the computation graph
of schedule lets say 1 at time $\ell$ then there exists a time $\ell'$
so that the computation graph under schedule $2$ is a superset of the
first computation graph. To be able to come to this conclusion we have
crucially used the fact that for an admissible schedule every section is
updated infinitely often. This shows that the performance under schedule
2 is at least as good as the performance under schedule 1.  Since the
roles of the schedules are symmetric, the claim follows.  
\end{proof}

\subsection{The EBP GEXIT Curve for Coupled Ensembles}
We come now to the key point, the computation of the EBP GEXIT curve.
From previous sections we have seen that FPs of forward DE are well defined
and can be computed by applying a parallel schedule. This procedure
allows us to compute {\em stable} FPs. As discussed in Section~\ref{sec:review}, it
was shown in \cite{MMRU09} how to compute unstable FPs for uncoupled ensembles
by a modified DE procedure in which the entropy is kept fixed and the channel
parameter is varied. The same procedure can be applied for coupled ensembles.
\begin{figure}[htp]
\begin{centering}
\input{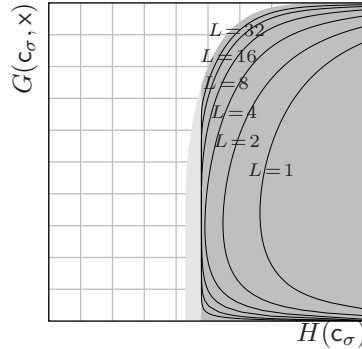}
\caption{EBP GEXIT curves of the ensemble $(\dl=3, \dr=6, L)$ 
for $L=1, 2, 4, 8, 16,$ and $32$ and transmission over the BAWGNC.
The BP thresholds are 
$\epsilon^{\BPsmall}(3, 6, 1)=0.66324$, 
$\epsilon^{\BPsmall}(3, 6, 2)=0.54701$, 
$\epsilon^{\BPsmall}(3, 6, 4)=0.49031$, 
$\epsilon^{\BPsmall}(3, 6, 8)=0.47928$, 
$\epsilon^{\BPsmall}(3, 6, 16)=0.47918$, 
$\epsilon^{\BPsmall}(3, 6, 32)=0.47917$, 
The light/dark gray areas mark the interior of the BP/MAP GEXIT function
of the underlying $(3, 6)$-regular ensemble, respectively.
}
\label{fig:lrLexit}
\end{centering}
\end{figure}
Figure~\ref{fig:lrLexit} shows the result of this numerical computation
when transmission takes place over the BAWGNC.  Note that the resulting
curves look very similar to the curves when transmission takes place
over the BEC, see \cite{KRU10}. For very small values of $L$ the curves are far to the right
due to the significant rate loss that is incurred at the boundary.
For $L$ around $10$ and above, the BP threshold of each ensemble is very
close to the MAP threshold of the underlying $(3, 6)$-regular ensemble,
namely $0.4792$. The picture strongly suggests that the same threshold saturation
effect occurs for general channels as it was shown analytically to hold
for the BEC in \cite{MMRU09}.

\section{A Possible Proof Approach}
So far the current discussion was empirical. Let us now quickly review which
parts of the proof in \cite{KRU10} can be extended easily and which currently seem difficult.\\
(i) {\em Constellation parameter:}
For the BEC, entropy is equal to the Bhattacharyya parameter which is equal to the erasure probability 
and which is also equal to twice the error probability. So, in this case
any of those values is a natural quantity to parametrize constellations.
For general channels, these parameters differ, and their choice is not necessarily equivalent
for the purpose of the proof.\\
(ii) {\em Existence of FP:}
Although we did not explicitly state it in this short paper, the
Existence Theorem~29 of \cite{KRU10}, which guarantees the existence
of a special FP of DE (c.f. Figure~\ref{fig:accordeonfp}) can be
extended to the general case by considering the Battacharyya
functional. The main technical difficulty in the proof arises due
to the fact that we are now operating on a space of symmetric
probability densities. So to extend the proof of the BEC, we need
to define appropriate metrics in this space so that we can apply
the necessary fixed-point theorems. Together with the use of extremes
of information combining methods, see \cite{RiU08}, the proof extends. \\
(iii) {\em Shape of the constellation and the transition
length:} A key ingredient in the proof for the BEC was to show that
any FP of DE has a very particular ``shape.'' More precisely, any FP
had a very ``fast'' transition between its extreme values.  Empirically,
for the general case we observe the same phenomena.  From
Figure~\ref{fig:accordeonfp} we see that the FP quickly saturates
to its maximum value (w.r.t.  physical degradation) of the stable FP
of the $(\dl, \dr)$-regular ensemble. To show this property analytically seems
currently to be one of the key difficulties in extending the proof.\\
(iv) {\em Construction of GEXIT Curve and the Area Theorem:}
Another key part of the BEC proof is the construction of a family
of FPs (not necessarily stable FPs). The GEXIT curve plus the fast
transition makes it possible in the case of the BEC to show that
the ``special'' FP which was constructed via the existence theorem,
has an associated channel parameter very close to the MAP threshold.
How to best construct the GEXIT curve is an open issue.
\begin{figure}[htp]
\begin{centering}
\input{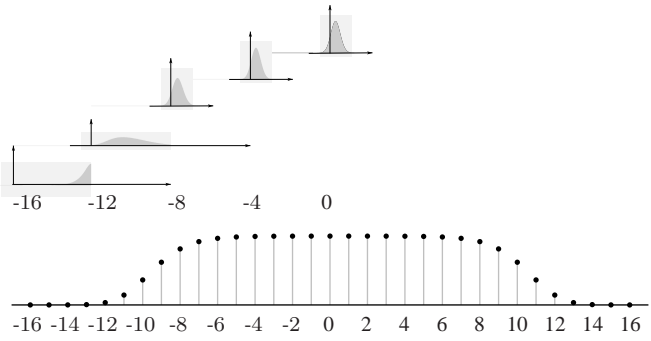}
\caption{Unimodal (special) FP of the $(\dl=3, \dr=6, L=16, w=3)$
ensemble for BAWGNC($\sigma$) with $\sigma=0.9480$ (channel entropy
$\approx 0.4792$).  
The bottom figure plots the entropy of the density at each section.
The values are small at the boundary and essentially constant in the middle. 
The top figure shows the actual densities at sections $0, \text{-}4, \text{-}8, \text{-}12, \text{-}16$. Notice that
the mass shifts towards the right as sections go to $\text{-}16$. Also
plotted in the top figure at section $0$ is the stable FP of DE for  
the  $(3,6)$-regular ensemble at $\sigma=0.9480$. The density is
right on the top of the one at section $0$.} \label{fig:accordeonfp} \end{centering}
\end{figure}

\section{How to Mitigate the Rate Loss} 
We have seen that by coupling ensembles we can increase the
threshold substantially. We also know that, due to the boundary condition,
we incur a rate loss (see Lemma~\ref{lem:designrate}). This rate loss
decreases linearly in $2 L+1$, the length of the chain. Therefore, by
picking $L$ large we can ensure that the rate loss is as small as desired.
But large $L$ implies large codeword lengths and also may require a
large number of iterations in the decoding process. This motivates  
to find ways of mitigating this rate loss.

To keep things simple, we consider transmission over the BEC. The same
techniques and trade-offs apply to general BMS channels, but of course
the given numerical values will change.
We discuss two basic techniques: (i) rather than setting all boundary
variables to be known, it suffices to set to known only a smaller
fraction; (ii) it suffices to start the process at only one boundary
rather than both. In addition there might be a benefit to consider
ensembles defined on a circle rather than a line. This symmetrizes all
positions of the ensemble, which in turn might lead to more efficient
implementations.

\subsection{Circular Ensembles}
Consider a ``circular'' ensemble.  This ensemble is defined in a similar
manner as the $(\dl, \dr, L, w)$ ensemble except that the positions
are now from $0$ to $K-1$ and the index arithmetic is performed modulo
$K$. This circular ensemble has design rate equal to $1-\dl/\dr$.
If we let $K=2L+w$ and if we set $w-1$ consecutive variable positions to $0$
then we recover the ensemble $(\dl, \dr, L, w)$. This in itself gives
a possibly more efficient way of implementing coupled ensembles. In this
implementation all positions are symmetric, except for the received values,
which are modified for the chosen $w-1$ consecutive positions.

Let us now generalize the construction. For $k \in [0, K-1]$ let $\kappa_k
\in [0, 1]$ denote the fraction of variable nodes at position $k$ which
we set to be known.  E.g., if we set $\kappa_0=\kappa_1=\cdots =\kappa_{\w-2}=1$
and all other $\kappa_i$ values to $0$ then we recover the previous case.
Define $\underline{\kappa}=\{\kappa_i\}$.  As we will see shortly, from
a rate perspective, it is not necessarily the best to set all variables
at a certain position to $0$.  Further, it can be better to choose the
``boundary'' positions to be non-consecutive. This is why it is useful to
introduce the above general model.

To start, let us compute
the design rate for the above set-up.  We denote this ensemble by $(\dl,
\dr, K, w, \underline{\kappa})$.
\begin{lemma}[Rate]
The design rate 
$R(\dl, \dr, K, w, \underline{\kappa})$
of the ensemble $(\dl, \dr, K, w, \underline{\kappa})$ is given by
\begin{align*}
1-\frac{\dl}{\dr} -\frac{\dl}{\dr} \frac{ 
\sum_{k=0}^{K-1} [\kappa_k- (\frac1{w} \sum_{j=0}^{w-1}
\kappa_{k-j})^{\dr}]}{\sum_{k=0}^{K-1} (1-\kappa_k)}.
\end{align*}
\end{lemma}
\begin{proof}
The design rate is equal to $1-C/V$, where $C$ and  
 $V$ are the number check and variable nodes 
which are not fixed a priori to $0$. Let us call those variable nodes ``free.''

Let us start with $V$. There are $M$ variable nodes per section.
A fraction $\kappa_i$ of those is permanently fixed to $0$. Therefore,
the number of free variable nodes is
$V=M \sum_{k=0}^{K-1} (1-\kappa_k)$. 

At each section there are $M \frac{\dl}{\dr}$ check nodes.  A check
node imposes a constraint on the free variable nodes if at least one
of its connections goes to a free variable node.  The probability that
a particular edge of a check node at position $k$ is connected to a
frozen bit is equal to $\frac1{w} \sum_{j=0}^{w-1} \kappa_{k-j}$,
where all index arithmetic is modulo $K$. This implies that the
probability that {\em all} $\dr$ edges of a check node at position $k$
are connected to frozen variables is equal to $(\frac1{w} \sum_{j=0}^{w-1}
\kappa_{k-j})^{\dr}$.  Therefore, $C=M K \frac{\dl}{\dr} (1- \frac1{K}
\sum_{k=0}^{K-1} (\frac1{w} \sum_{j=0}^{w-1}
\kappa_{k-j})^{\dr})$.
\end{proof}

\begin{example}[Contiguous and Uniform Boundary]
Consider the $(\dl=3, \dr=6, K, w, \underline{\kappa})$ ensemble where
we set $\kappa_0=\kappa_1=\dots=\kappa_{w-2}=\kappa$. We set all other
values $\kappa_k$ equal to $0$.
For the choice $\kappa=1$ we know that we can achieve a threshold of
roughly $0.48815$ irrespective of the length of $K$. What happens if we
pick $\kappa$ strictly less than $1$?

Let $\delta$ denote the ``effective'' erasure probability at the boundary.
I.e., $\delta$ denotes the fraction of variables in the boundary which
are free and which were erased during the transmission process. We have
$\delta=(1-\kappa) \epsilon$. What is the threshold $\epsilon^{\BPsmall}$
that can be achieved (for arbitrary large $K$) for a given value of $\delta$?

Figure~\ref{fig:rateloss} shows the plot of $\epsilon^{\BPsmall}$
according to DE for $w=3, 4, 8, 16$ and $32$ as a function of $\delta$.
For e.g. $w=3$, up to $\delta=0.23$ the achievable threshold is still
equal to its maximal value, namely $0.48815$. For higher values of $\delta$
the threshold gracefully decreases.  For $\delta=0.23$ we have $\kappa =
1-0.23/0.48815 = 0.529$. For this value of $\kappa$ the corresponding rate
for $K=25$ is $0.478$. This is considerably larger than $0.4604$, which
is the rate if  $\kappa=1$. Indeed, the rate loss has almost been halved.

\begin{figure}[htp] \centering
\input{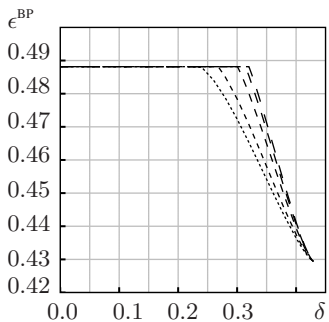} 
\caption{\label{fig:rateloss} 
Behavior of the threshold $\epsilon^{\BPsmall}$ for a uniform and
consecutive boundary condition as a function of the effective erasure
probability $\delta$ at the boundary. The parameters are $w=3, 4,
8, 16$ and $32$.  For $w\!=\!3, 4$ we choose $K\!=\!200$ and for for $w=8,
16, 32$ we chose $K\!=\!400$.  The $x$-axis shows $\delta=(1\!-\!\kappa)
\epsilon^{\BPsmall}$.  The $y$-axis shows the achievable threshold
$\epsilon^{\BPsmall}$.  The curve for $w\!=\!3$ is the left-most curve and the
curve for $w=32$ is the right-most curve.  We have $\epsilon^{\BPsmall}
\!\approx\! 0.48815$ up to $\delta \!\approx\! 0.23$ for $w\!=\!3$, up to $\delta
\!\approx\! 0.267$ for $w\!=\!4$, up to $\delta \approx 0.3$ for $w\!=\!8$, up
to $\delta \approx 0.31$ for $w=16$, and up to $\delta \!\approx\! 0.32$
for $w=32$.  For larger values of $\delta$, $\epsilon^{\BPsmall}$
gracefully decreases to $\epsilon^{\BPsmall}\!\approx\! 0.4294$.} \end{figure}
\end{example}

\begin{example}[Non-Contiguous and Non-Uniform Boundary]
We can do slightly better. Pick $\delta_0=0.22$ and $\delta_2=0.30$. Note
that these two positions are non-contiguous.  For this choice we still get
a threshold of $0.48815$. The corresponding rate for $K=25$ is $0.4806$,
which is slightly higher than in the previous case.  \end{example}

\subsection{One-Sided Ensemble}
An alternative scheme is to define the ensemble on a line but to
employ different terminations at the two boundaries. To be precise.
Let the variable nodes be positioned from $0$ to $K-1$.  Assume the
usual $(\dl, \dr, w)$ case. At the "right" boundary use the following
termination scheme.  At position $K$ there are $M \frac{\dl}{\dr}
\frac{w-1}{2}$ check nodes (instead of the usual $M \frac{\dl}{\dr}$).
Any edge, which under the standard connection rules should connect
to a check node at a position $K$ or larger is mapped to a check node
socket at position $K$. In expectation, exactly $M \dl \frac{w-1}{2}$
such sockets are needed so that all check nodes at position $K$ have
degree exactly $\dr$.  Therefore, locally the right-hand side boundary
behaves exactly like a $(\dl, \dr)$ ensemble and there is no rate loss
associated with this boundary.  

On the left-hand side we proceed as in the standard ensemble. This
reduces the rate-loss by a factor 2 compared to the standard ensemble.
E.g., for $K=25$ we get for our usual $(3, 6)$-ensemble a rate of $0.48$.
But we can do better.  

\begin{example}[One-Sided Termination of $(3, 6)$ Ensemble]
Let $w=3$ and consider an ensemble in which the right boundary is
terminated without rate loss as described above.  Under the standard
scheme, the check nodes at position $0$ have degree $2$ and check nodes
at position $1$ have degree $4$.

Take a fraction $\alpha$ of the check nodes at position $0$ and merge
them with check nodes at position $1$. Those merged nodes at position
$1$ have degree $6$ as in the regular case. As long as $\alpha$ is sufficiently small
the threshold still remains unchanged. But this merging further reduces 
the incurred rate loss.
\end{example}

A small note of caution might be in order at this place. Even though
we can, as we just showed, mitigate the rate loss, this comes at some
price -- the number of required iterations will go up. A characterization
of the involved trade off would be of high practical interest.

\section{Conclusion}
Starting with the work by Felstr\"{o}m and Zigangirov \cite{FeZ99}, it
has been known that coupled ensembles have an excellent performance.
This was confirmed via threshold computations by Sridharan, Lentmaier,
Costello and Zigangirov for the BEC \cite{SLCZ04} and by Lentmaier,
Sridharan, Zigangirov and Costello for general channels \cite{LSZC05}.
In \cite{KRU10} it was shown that for transmission over the BEC 
the BP behavior of coupled ensembles is essentially equal to the
MAP behavior of the underlying ensemble.
 
The current paper provides numerical evidence which suggest that
exactly the same behavior occurs also for transmission over general
BMS channels.  We have further extended some of the basic techniques
and statements which were used in \cite{KRU10} to accomplish the
proof for the BEC to the general setup. A complete proof is
unfortunately still open. As discussed already in \cite{KRU10}, such
a proof would automatically show that it is possible to achieve 
capacity under iterative coding, and in addition, the convergence
to capacity would be uniform over the whole class of BMS channels.

\section{Acknowledgments}
SK acknowledges support of NMC via the NSF collaborative grant CCF-0829945
on ``Harnessing Statistical Physics for Computing and Communications''. SK
would also like to thank Misha Chertkov for his encouragement.

\bibliographystyle{IEEEtran} 
\bibliography{lth,lthpub}
\end{document}